\documentclass[runningheads,12pt]{llncs}
\usepackage{graphicx}
\usepackage{tikz}
\usetikzlibrary{snakes}
\usepackage{xcolor}
\usepackage{hyperref}
\usepackage{xspace}
\usepackage[left=4.32cm,top=4.32cm,right=4.32cm,bottom=4.32cm]{geometry}
\usepackage{amssymb,amsmath,amsfonts,graphicx}
\usepackage{subcaption}

\newcommand\CC{\mathcal{C}}
\newcommand\HH{\mathcal{H}}
\newcommand\Z{\mathbb{Z}}
\newcommand\ILT{\textrm{ILT}}

\usepackage{marginnote}
\usepackage{verbatim}

\usepackage{sectsty}
\allsectionsfont{\sffamily}

\begin{document}

\title{The Iterated Local Directed Transitivity Model for Social Networks\thanks{The first author acknowledges funding from an NSERC Discovery grant, while the third author acknowledges support from an NSERC Postdoctoral Fellowship.}}
\titlerunning{The ILDT Model for Social Networks}

\author{Anthony Bonato\inst{1} \and Daniel W.\ Cranston\inst{2} \and Melissa A.\ Huggan\inst{1} \and Trent Marbach\inst{1} \and Raja Mutharasan\inst{1}}

\authorrunning{A. Bonato, D.W.\ Cranston, M.A.\ Huggan, T. Marbach, R. Mutharasan}

\institute{Ryerson University, Toronto, ON, Canada \\
\email{\href{mailto:abonato@ryerson.ca}{abonato@ryerson.ca},
\href{mailto:melissa.huggan@ryerson.ca}{melissa.huggan@ryerson.ca}, \href{mailto:trent.marbach@ryerson.ca}{trent.marbach@ryerson.ca}, \href{mailto:rmutharasan@ryerson.ca}{rmutharasan@ryerson.ca}}
\and Virginia Commonwealth University, Richmond, VA, USA \\
\email{\href{mailto:dcranston@vcu.edu}{dcranston@vcu.edu}}}

\maketitle

\begin{abstract}

We introduce a new, deterministic directed graph model for social networks, based on the transitivity of triads. In the Iterated Local Directed Transitivity (ILDT) model, new nodes are born over discrete time-steps and inherit the link structure of their parent nodes. The ILDT model may be viewed as a directed graph analog of the ILT model for undirected graphs introduced in \cite{ilt}. We investigate network science and graph-theoretical properties of ILDT digraphs. We prove that the ILDT model exhibits a densification power law, so that the digraphs generated by the models densify over time. The number of directed triads are investigated, and counts are given of the number of directed 3-cycles and transitive $3$-cycles. A higher number of transitive 3-cycles are generated by the ILDT model, as found in real-world, on-line social networks that have orientations on their edges. We discuss the eigenvalues of the adjacency matrices of ILDT digraphs. We finish by showing that in many instances of the chosen initial digraph, the model eventually generates digraphs with Hamiltonian directed cycles.

\end{abstract}

\section{Introduction}

Real-world, complex networks contain numerous mechanisms governing link formation. \emph{Balance theory} (or \emph{structural balance theory}) in social network analysis cites several mechanisms to complete triads (that is, subgraphs consisting of three nodes) in social networks \cite{ek,he}. A central mechanism in balance theory is \emph{transitivity}: if $x$ is a friend of $y,$ and $y$ is a friend of $z,$ then $x$ is a friend of $z$; see, for example, \cite{scott}. Directed networks of ratings or trust scores and models for their propagation were first considered in \cite{guha}. \emph{Status theory} for directed networks, first introduced in \cite{lhk}, was motivated by both trust propagation and balance theory. While balance theory focuses on likes and dislikes, status theory posits that a directed link indicates that the creator of the link views the recipient as having higher status. For example, on Twitter or other social media, a directed link captures one user following another, and the person they follow may be of higher social status. Evidence for status theory was found in directed networks derived from Epinions, Slashdot, and Wikipedia \cite{lhk}. For other applications of status theory and directed triads in social networks, see also \cite{shc,sm}.

The \emph{Iterated Local Transitivity} (\emph{ILT}) model introduced in \cite{ilt1,ilt} and further studied in \cite{ilm,mason}, simulates social networks and other complex networks. Transitivity gives rise to the notion of \emph{cloning}, where a node $x$ is adjacent to all of the neighbors of $y$. Note that in the ILT model, the nodes have local influence within their neighbor sets. Although the model graph evolves over time, there is still a memory of the initial graph hidden in the structure.  The ILT model simulates many properties of social networks. For example, as shown in \cite{ilt}, graphs generated by the model densify over time and exhibit bad spectral expansion. In addition, the ILT model generates graphs with the small-world property, which requires the graphs to have low diameter and dense neighbor sets.

In the present work, we introduce a directed analog of the ILT model, where nodes are added and copy the in- and out-neighbors of existing nodes. The model simulates link creation in social networks, where new actors enter the network, and directed edges are added via transitivity through the lens of status theory. For example, in link formation in a directed social network such as Twitter, a new user may reciprocally follow an existing one, then in turn follow their followers. We consider the simplified setting where new nodes copy all of the links of their parent node. Our model, and iterated models more generally \cite{ilm}, provide counterparts for earlier studied random models for complex networks involving copying \cite {kumar} or duplication \cite{dupl}.

More formally, the \emph{Iterated Local Directed Transitivity} (\emph{ILDT}) model is deterministically defined over discrete time-steps as follows. The only parameter of this deterministic model is the initial digraph $G=G_{0}$. For a non-negative integer $t$, the graph $G_{t}$ represents the digraph at time-step $t$. Suppose that the directed graph $G_{t}$ has been defined for a fixed time $t\geq 0$. To form $G_{t+1}$, for each $x \in V(G_{t}$), add a new node $x'$ called the \emph{clone} of $x$. We refer to $x$ as the \emph{parent} of $x',$ and $x'$ as the \emph{child} of $x.$ Between $x$ and $x'$ we add a bidirectional arc, representing a reciprocal status (or friendship) relationship between them. For arcs $(x,z)$ and $(y,x)$ in $G_{t}$, we add arcs $(x',z)$ and $(y,x'),$ respectively, in $G_{t+1}$. See Figure~1. We refer to $G_t$ as an \emph{ILDT digraph}. Note that the clones form an independent set in $G_{t+1}$.
\begin{figure}[h!]
\begin{center}
\includegraphics[scale=1]{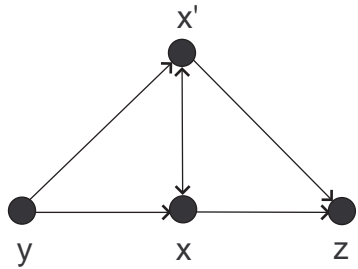}
\caption{Adding a clone $x'$ in ILDT.}\label{animals0}
\end{center}
\end{figure}

In Figure~\ref{ex: ILDT_model}, we illustrate several time-steps of the ILDT model beginning with the directed 3-cycle.
\begin{figure}[h!]
\begin{center}
\includegraphics[scale=0.3]{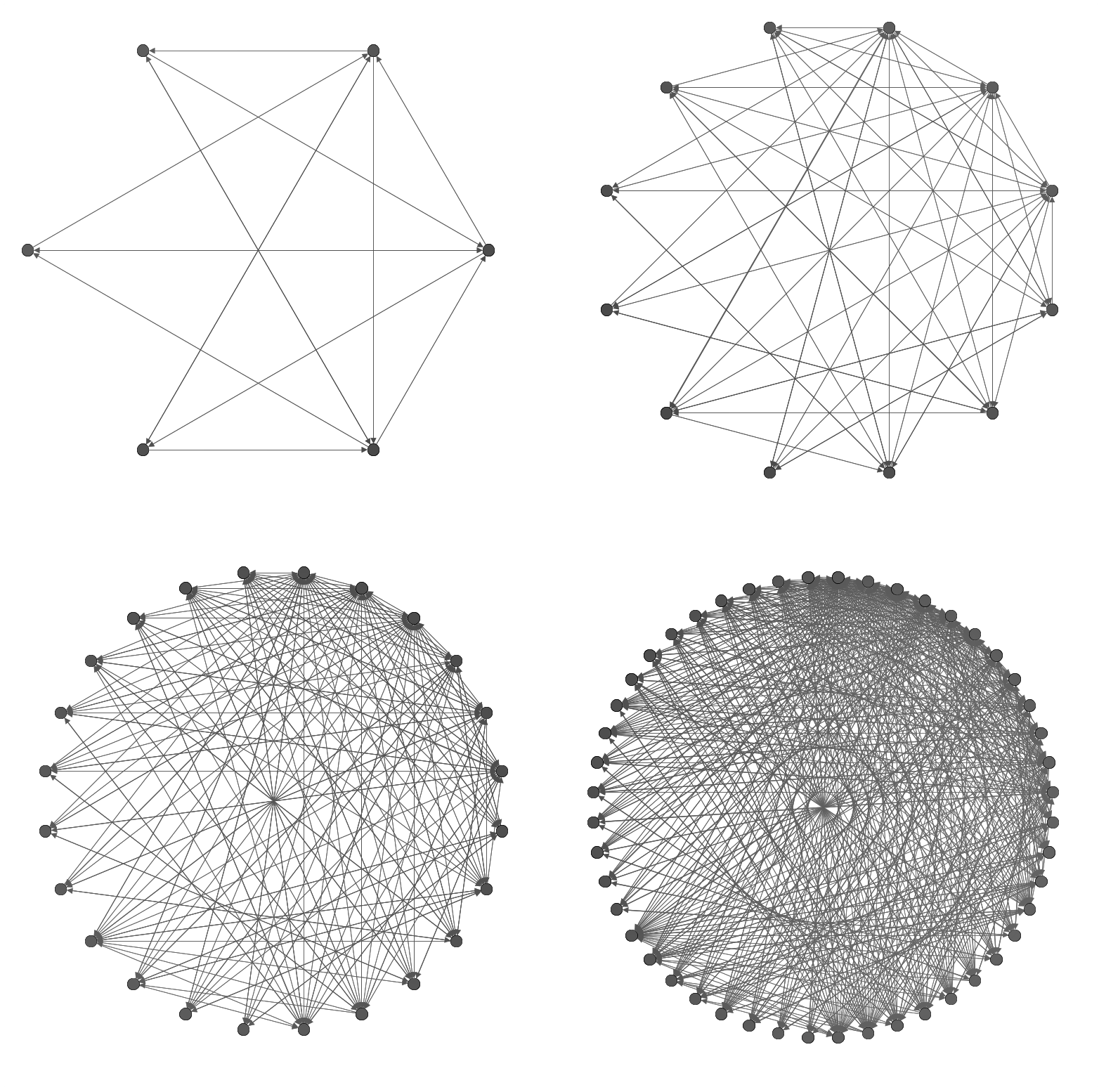}
\end{center}
\caption{The ILDT graphs $G_t,$ with $t=1,2,3,4$ of the ILDT model, where the initial graph is the directed 3-cycle.}\label{ex: ILDT_model}
\end{figure}

The paper is organized as follows. In Section~\ref{sdense}, we prove that the ILDT model exhibits a densification power law, so that the digraphs generated by the model densify over time. The number of directed triads are investigated in Theorem~\ref{count}, and precise counts are given of the number of directed 3-cycles and transitive cycles. These counts are contrasted, and it is shown that the transitive cycles are more abundant (as is the case with social networks; see \cite{lhk}). We include a discussion of the eigenvalues of the adjacency matrices of ILDT directed graphs.
In Section~\ref{seigen}, the eigenvalues of ILDT directed graphs are investigated.
In Section~\ref{scycle}, we explore directed cycles of larger order in ILDT graphs. We show that ILDT digraphs are acyclic if the initial digraph is such, and that for many instances of initial digraphs, the model eventually generates graphs with Hamiltonian directed cycles. We finish with open problems.

For a general reference on graph theory, the reader is directed to \cite{west}. For background on social and complex networks, see \cite{bbook,bt,chung1}. Throughout the paper, we consider finite, directed graphs with bidirectional edges allowed. We refer to a directed edge as an \emph{arc}. For nodes $x$ and $y$ of a graph, if there is an arc between $x$ and $y$, we denote it by $(x,y)$; if there is a \emph{bidirectional arc} between $x$ and $y$ we denote such an arc with usual graph theoretic notation of $xy$. When counting the number of arcs within a graph, bidirectional arcs each contribute $2$ to the final count. We use $\log n$ to be the logarithm of $n$ in base 2.

\section{Densification and triad counts}\label{sdense}

As we referenced in the introduction, social networks \emph{densify}, in the sense that the ratio of their number of arcs to nodes tends to infinity over time \cite{les1}. We show in this section that the ILDT model always generates digraphs that densify, and we give a precise statement below of its densification power law. As a phenomenon specific to digraphs, we consider the differing counts of directed and transitive 3-cycles in graphs generated by the model.

The number of nodes of $G_t$ is denoted $n_t$, the number of arcs is denoted $e_{t}$, and the number of bidirectional arcs is denoted  $b_{t}$. Note that $e_t$ contains two arcs for each bidirectional arc. We establish elementary but important recursive formulas for these parameters.

\begin{lemma}\label{thm: node and edge counts}
Let $G_{0}$ be a digraph with $n_{0}$ nodes, $e_{0}$ arcs, and $b_0$ bidirectional arcs. For all $t\ge 1,$ we have the following:
\begin{enumerate}
\item $n_{t} = 2^{t}n_{0}$;
\item $e_{t} = 3e_{t-1} + 2n_{t-1}$; and
\item $b_{t} = 3b_{t-1} + n_{t-1}.$
\end{enumerate}
\end{lemma}

\begin{proof}
Item (1) follows immediately as the number of nodes doubles in each time-step. For item (2), for each node in $G_{t-1}$ with $t>0,$ after cloning there will be a bidirectional arc between each parent and their child. These count as $2n_{t-1}$-many arcs. For every arc $(x,y)$ in $G_{t-1}$, arcs $(x,y')$ and $(x',y)$ appear in $G_{t}.$ Hence, for every arc in $G_{t-1}$, three arcs of $G_{t}$ are generated. Summing these two counts gives the desired expression for $e_t.$ Item (3) follows analogously to (2), except that we count the bidirectional edges once; hence, there are $n_{t-1}$-many bidirectional arcs.\qed 
\end{proof}

We now state the densification power law for ILDT graphs. For positive integer-valued functions $f_t$ and $g_t$, we use the expression $f_t \sim g_t$ if ${f_t}/{g_t}$ tends to $1$ as $t$ tends to $\infty.$

\begin{corollary}\label{cc}
In the ILDT model, we have that
$$
\frac{e_{t}}{n_{t}} \sim \left(\frac{3}{2}\right)^{t} \frac{(e_0+2n_0)}{n_0}.
$$
In particular, we have that $e_t \sim C \cdot (n_t)^a,$ where $a = \log 3$ and $C=\frac{e_0+2n_0}{(n_0)^a}$.
\end{corollary}

\begin{proof}
By Lemma~\ref{thm: node and edge counts}, we have that
\begin{eqnarray*}
e_{t} &=&3^{t}e_{0} + 3^{t-1} 2^{1}n_{0} +  3^{t-2} 2^{2}n_{0}+\ldots+3^12^{t-1}n_{0}+ 2^{t}n_{0}\\
&=& 3^{t}e_{0} + 3^{t-1} 2n_{0}\left(\frac{1-\left(\frac{2}{3}\right)^{t}}{1-\frac{2}{3}}\right)\\
&= &3^{t}\left(e_{0} + 2  n_{0}\right) -  2^{t}(2 n_{0}).\\
\end{eqnarray*}

We then derive that
$$
\frac{e_{t}}{n_{t}} = \frac{3^{t}\left(e_{0} + 2 \cdot n_{0}\right) -  2^{t}(2 n_{0})}{2^{t}n_{0}} \sim \left(\frac{3}{2}\right)^{t}\frac{(e_0+2n_0)}{n_0},
$$
and the result follows.\qed
\end{proof}

We next consider $3$-node subgraph counts in the ILDT model, where we find a higher number of transitive $3$-cycles relative to directed 3-cycles.
A similar phenomenon was found in directed network samples in social media such as Epinions, Slashdot, and Wikipedia, where transitive 3-cycles appear much more commonly than directed 3-cycles; see \cite{lhk}.
If the initial graph has no directed 3-cycles, then our results show there are none at any time-step in the evolution of the model.
For nodes $x$, $y$, $z$, if there exists a cycle with arcs $(x, y)$, $(y, z)$, and $(z, x),$ then we abbreviate this directed $3$-cycle to $(x, y, z)$. We take the directed $3$-cycles $(x,y,z)$ and $(y,z,x)$ to be the same cycle. For a $3$-cycle where the arcs $(x, y)$, $(y, z)$, and $(x, z)$ are present, we denote this transitive $3$-cycle by $xyz$. A \emph{bidirectional $3$-cycle} is one that consists of three bidirectional arcs.
Although, strictly speaking, a bidirectional $3$-cycle contains six transitive $3$-cycles and two directed $3$-cycles, 
 we will distinguish these by asserting that each transitive and directed $3$-cycle must contain at least one non-bidirectional arc.

\begin{theorem}\label{thm: directed and transitive cycle count}
In the graph $G_t$, let $D_{t}$ be the number of directed $3$-cycles, $T_{t}$ be the number of transitive $3$-cycles, and $B_t$ be the number of bidirectional $3$-cycles. We then have that

\begin{enumerate}
\item $D_{t+1} = 4D_{t}$;
\item $T_{t+1} = 4T_{t} + 4(e_{t}-2b_t)$; and
\item $B_{t+1} = 4B_t + 2b_t.$
\end{enumerate}
\end{theorem}

\begin{proof}
For item (1), consider a directed $3$-cycle in a graph $G_{t}$, labeled $(a,b,c)$. Each node can be replaced by its clone to produce a new directed $3$-cycle; hence, $(a', b, c)$, $(a,b',c)$, and $(a,b,c')$ are all directed $3$-cycles generated by the directed  $3$-cycle $(a,b,c)$ from $G_{t}$. Thus, for each $3$-cycle in $G_{t}$, there are four directed $3$-cycles in $G_{t+1}$. Therefore, $D_{t+1} \ge 4D_{t}$.

To establish the upper bound, suppose, by way of contradiction, that there exists another directed $3$-cycle in $G_{t+1}$ which was not previously counted. Such a directed $3$-cycle cannot involve only parent nodes (since it would be from $G_{t}$) and it also cannot involve two clones because clones form an independent set. Hence, it must involve two nodes from $G_{t}$ and one of the clones, call it $(a,d',c)$. But since $d'$ has all the same adjacencies as $d$, this implies that $(a,d,c)$ is a directed $3$-cycle which is in $G_{t}$, a contradiction, or that $a,d,c$ are not all distinct. In the later case, we can assume that $a=d$, and so the directed $3$-cycle includes arcs $(a,c)$ and $(c,a')$, implying that $3$-cycle is bidirectional, which is not counted in $D_{t+1}$. Hence, $D_{t+1} = 4D_t.$

To prove item (2), observe that for every non-bidirectional arc $(x,y)$ in $G_{t}$, there will be four transitive $3$-cycles in $G_{t+1}$ formed with their clones using bidirectional arcs: $xx'y$, $x'xy$, $xyy'$, and $xy'y$. Thus, we add $4(e_{t}-2b_t)$ to the count of $T_{t+1}$ (note that $e_{t}-2b_t$ counts the number of non-bidirectional arcs)

Existing transitive $3$-cycles from $G_{t}$ also exist in $G_{t+1},$ and substituting a clone for its parent will produce a new transitive $3$-cycle. Thus, each original transitive $3$-cycle in $G_{t}$ gives four transitive $3$-cycles in $G_{t+1}$ (namely, the cycle itself, and three produced from clone substitution). It is straightforward to check that these are the only transitive 3-cycles in $G_{t+1}.$ Therefore, $T_{t} + 3T_{t}$ contributes to the total count of transitive $3$-cycles. We then have that $T_{t+1} = 4T_{t} + 4(e_{t}-2b_t).$

Item (3) can be shown similarly, as each bidirectional $3$-cycle of $G_t$ will correspond to four bidirectional $3$-cycles in $G_{t+1}$. Each bidirectional arc $xy$ in $G_{t}$ will correspond to two unique bidirectional $3$-cycles in $G_{t+1}$ formed from the nodes $x,x',y,y'$. \qed
\end{proof}

We now present an exact expressions for $D_{t},$ $T_{t},$ and $B_t.$

\begin{theorem}\label{count}
In the ILDT digraph $G_t$, we have that
\begin{enumerate}
\item $D_{t} = 4^{t}D_{0}$;
\item $T_{t} = 4^t T_0 + 4 (4^{t}-3^t) (e_0 - 2b_0)$; and
\item $B_t = 4^t B_0 + 2b_0 (4^t - 3^t) + n_0 (4^t - 2\cdot3^t+2^t).$
\end{enumerate}
\end{theorem}

\begin{proof}
Item (1) follows from Theorem~\ref{thm: directed and transitive cycle count} (1) by induction. For (2), by Theorem~\ref{thm: directed and transitive cycle count} (2), we derive that
\begin{align*}
T_{t} &= 4^{t}T_{0}+ \sum_{i=1}^{t} 4^i (e_{t-i} - 2b_{t-i}).
\end{align*}
By the proof of Corollary~\ref{cc}, along with a similar argument for $b_t$, we have that
\begin{align}
e_{t-i} &= 3^{t-i}(e_{0} + 2n_{0}) - 2^{t-i+1}n_{0}, \nonumber \\
b_{t-i} &= 3^{t-i}(b_{0} + n_{0}) - 2^{t-i}n_{0}, \label{expression for edges2} \text{ and so} \\
e_{t-i} -2b_{t-i} &= 3^{t-i}(e_{0} -2b_0). \label{expression for edges3}
\end{align}
From \eqref{expression for edges3}, we derive that
$$
T_{t} = 4^{t} T_{0}+ \sum_{i=1}^{t} 4^i (3^{t-i} (e_0 - 2b_0))
$$
and item (2) follows by summing the geometric series.

For item (3), we have that $B_{t} = 4^{t}B_{0}+ \frac{1}{2} \sum_{i=1}^{t} 4^i b_{t-i}.$ From equation \eqref{expression for edges2}, we find (in a way analogous to $T_t$) the desired expression for $B_t.$  \qed \end{proof}

We consider next the ratio of $D_t$ and $T_t,$ which gives more explicit estimates on the relative abundance of transitive versus directed 3-cycles. By Theorem~\ref{count}, we have that
$$
\frac{D_{t}}{T_{t}} =\frac{4^{t}D_{0}}{ 4^t T_0 + 4 (4^{t}-3^t) (e_0 - 2b_0)} \sim  \frac{D_{0}}{T_{0} +  4 (e_0 - 2b_0)}.
$$
Hence, the ratio $\frac{D_{t}}{T_{t}}$ may be made as small as we like by choosing an appropriate initial digraph.

\section{Eigenvalues}\label{seigen}

We next consider eigenvalues of the adjacency matrices of ILDT digraphs. Spectral graph theory is a well-developed area for undirected graphs (see \cite{chung}) but less so for directed graphs (where the eigenvalues may be complex numbers with non-zero imaginary parts). We observe that if $G_t$ has adjacency matrix $A$, then $G_{t+1}$ has the following adjacency matrix:
\[
\left(\begin{matrix}
A & A+I \\
A+I & 0
\end{matrix}\right),
\]
where $I$ and $0$ are the appropriately sized identity and zero matrices, respectively. The following recursive formula (analogous to the one in the ILT model) allows us to determine all the eigenvalues of ILDT graphs from the spectrum of the initial graph. We have the following theorem from \cite{ilt}.

\begin{theorem}\label{ilt} \label{thm:NewEigenvalues}
Let $t \ge 0.$ If $\rho$ is an eigenvalue of the adjacency matrix of $G_t$, then the eigenvalues of the adjacency matrix of $G_{t+1}$ are
\[
\frac{\rho \pm \sqrt{\rho^2 +4(\rho+1)^2}}{2}.
\]
\end{theorem}

It is of interest to consider properties of the distribution of eigenvalues arising from ILDT digraphs graphs in the complex plane. For this, we consider the special case of $G_0$ a directed $3$-cycle, which has eigenvalues the 3rd roots of unity. The resulting eigenvalue distribution of these ILDT digraphs suggests a rich structure. We plot the eigenvalues corresponding to $G_t$ for $1 \leq t \leq 5$ in Figure~\ref{fig:unNorm}.

\begin{figure} \label{fig:complexWholeSmall} \centering
\subcaptionbox{ \label{fig:unNorm}  Standard.}{\includegraphics[scale=0.35]{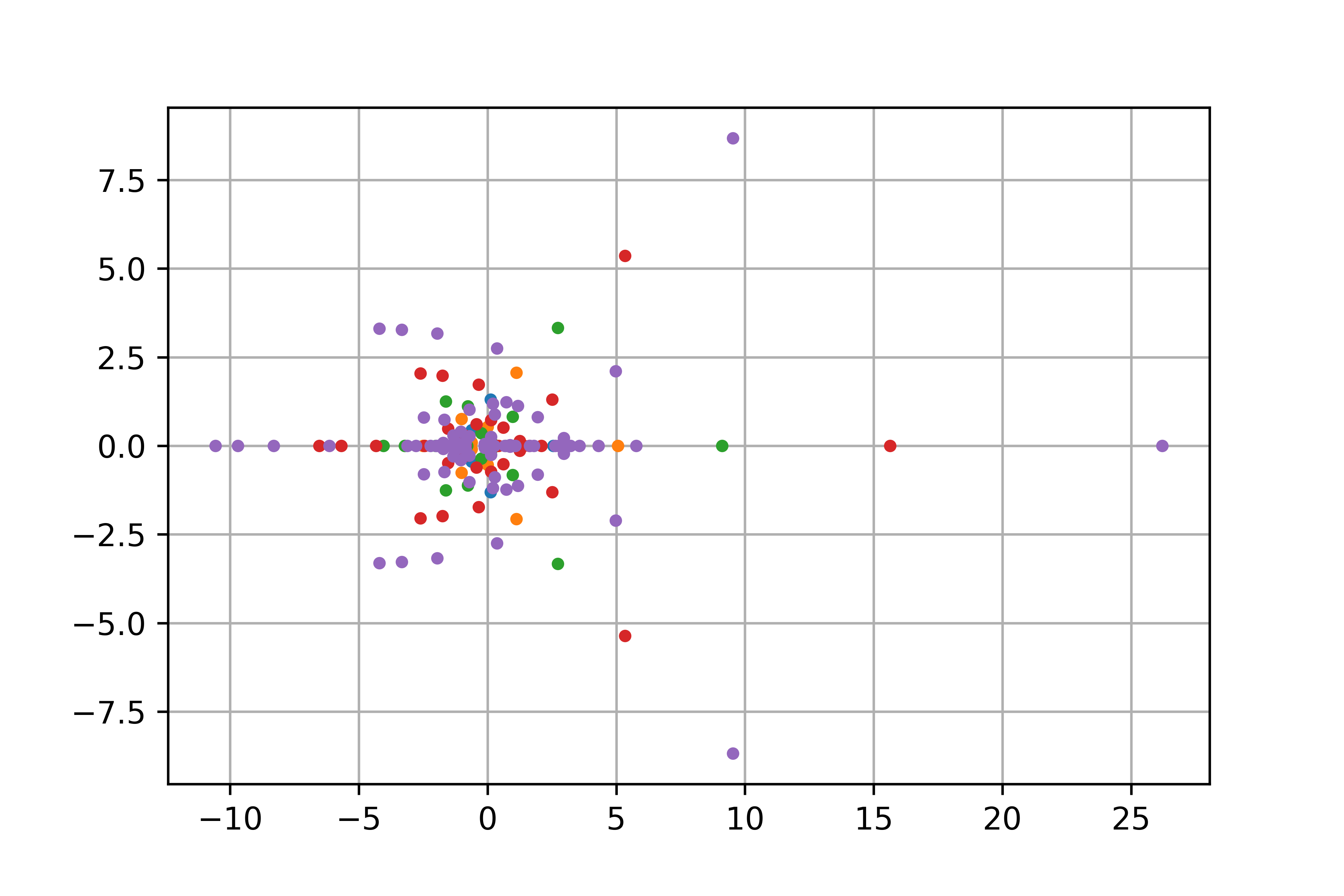}}
\subcaptionbox{ \label{fig:Norm} Normalized.} {\includegraphics[scale=0.35]{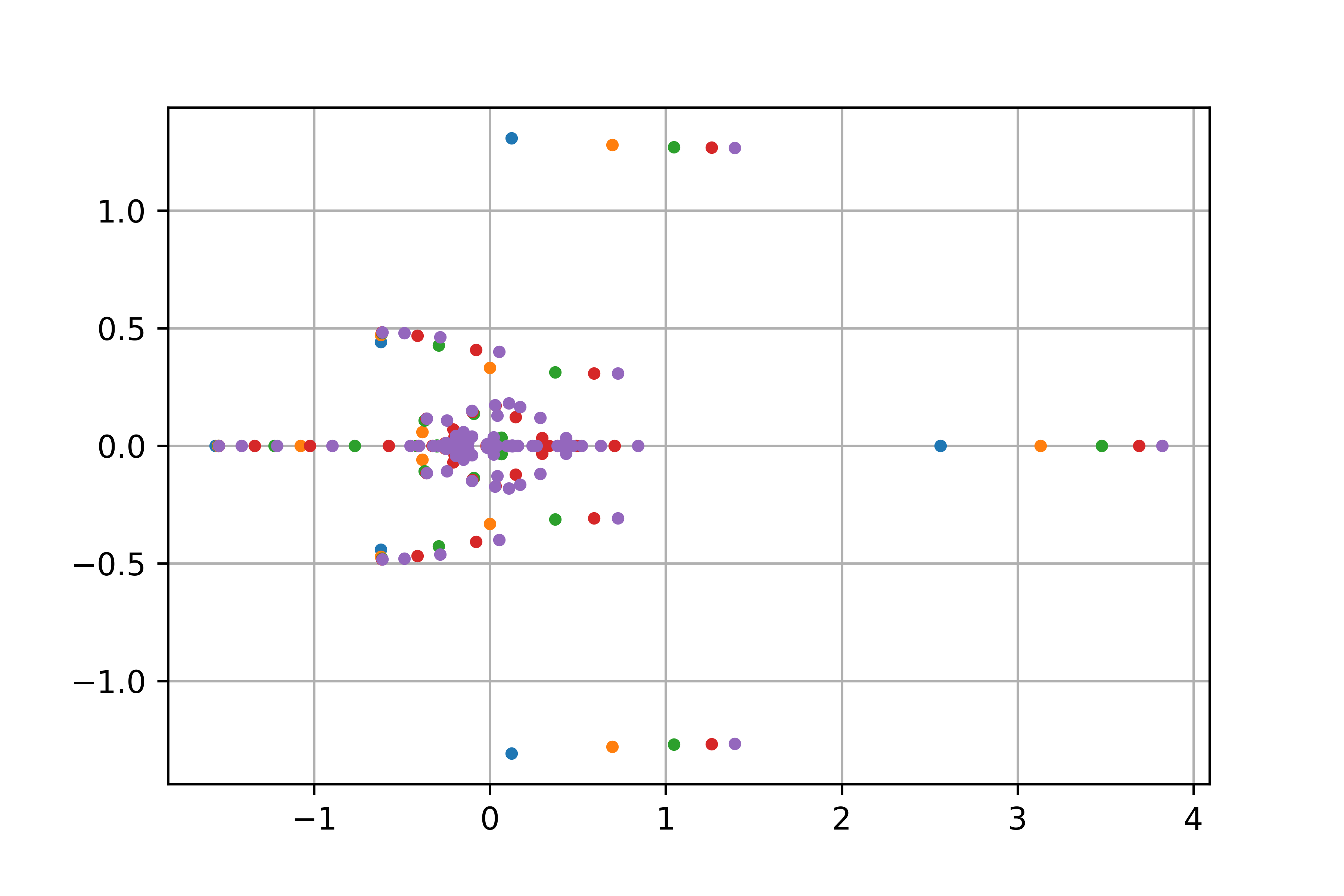}}
\subcaptionbox{ \label{fig:NormCurve} Normalized, with curve $C_t$.} {\includegraphics[scale=0.30]{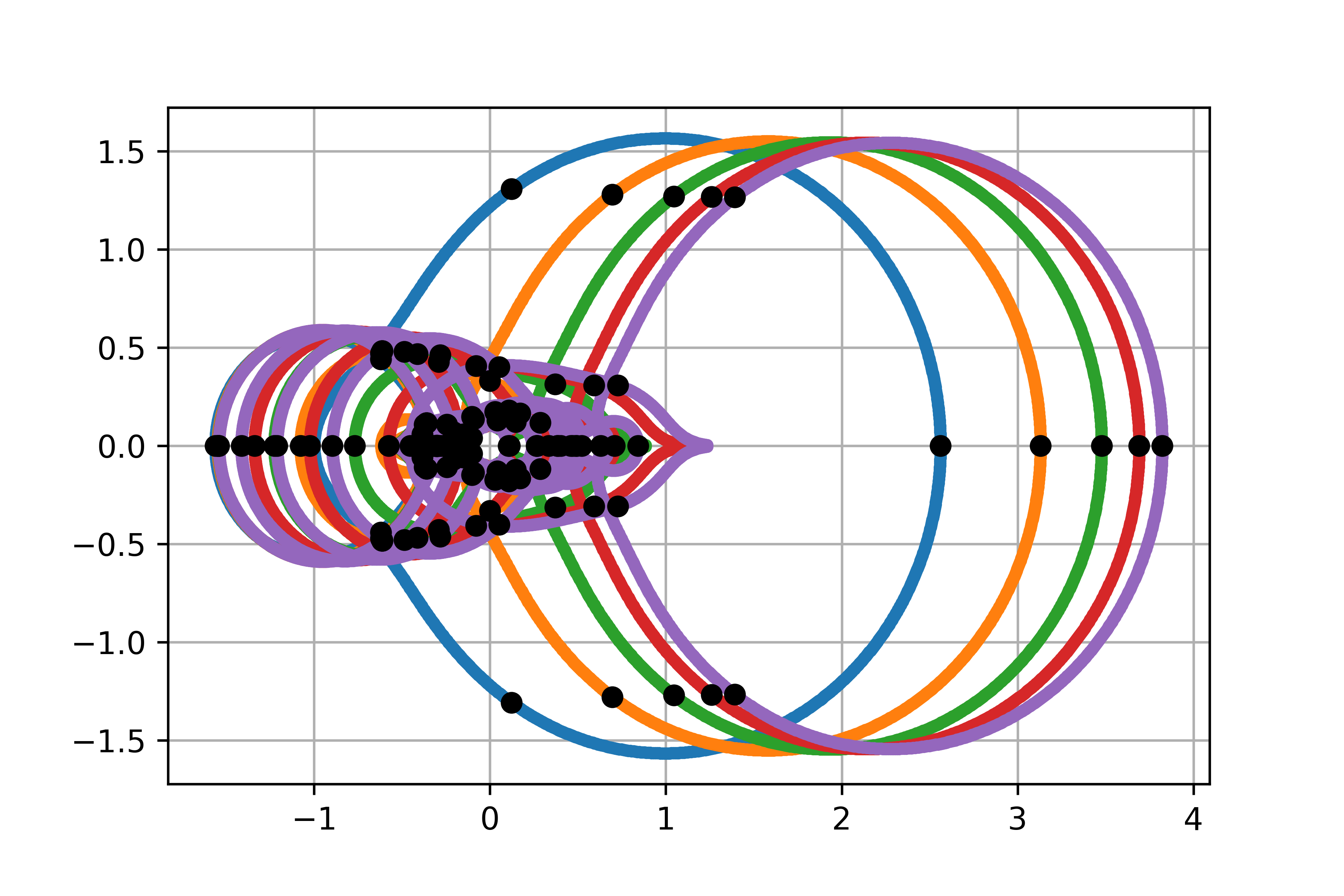}}
\caption{Eigenvalues in the complex plane of ILDT digraphs $G_t$, where $1 \leq t \leq 5$. Colors distinguish the time-steps. Figures (a) depicts the eigenvalues, (b) the normalized eigenvalues, and (c) depicts the curves $C_t.$}\label{animals} \end{figure}

If $\rho$ is an eigenvalue of the adjacency matrix of $G_t$ and $\rho$ is large in magnitude, then there is an eigenvalue of the adjacency matrix of $G_{t+1}$ that is approximately equal to $((1+\sqrt{5})/2) \rho$, and in a similar way there is an eigenvalue of the adjacency matrix of $G_{t+\alpha}$ that is approximately equal to $((1+\sqrt{5})/2)^\alpha \rho$. We \emph{normalize} the eigenvalues corresponding to $G_t$  by dividing them by $((1+\sqrt{5})/2)^t$ for $1 \leq t \leq 5$ in Figure~\ref{fig:Norm}.

Let $C_0$ be the circle in the complex plane of radius 1 centered at the origin. By applying the function $f(z) = (z \pm \sqrt{z^2 +4(z+1)^2})/2$ iteratively $t$ times to the points of $C_0$, we obtain a curve $C_t$ in the complex plane. If we let $G_0$ be the directed $n$-cycle, then the eigenvalues of $G_t$ lie on $C_t$. In Figure~\ref{fig:NormCurve}, we include $C_t$ and the eigenvalues of $G_t$ after normalization by dividing them by $((1+\sqrt{5})/2)^t$, where $G_0$ is the directed $3$-cycle. The curve $C_t$ was plotted after normalization for $t \leq 30$, and there was no noticeable difference between the time-steps $t=15$ and $t=30$. The structure after 30 iterations is provided in Figure~\ref{fig:30it}. We suspect that as $t$ approaches infinity, the normalization of $C_t$ approaches a specific, limiting curve. As a result, the normalized mapping applied to the $n$th roots of unity would approach limiting points, and these limiting points can be calculated from the curve.

\begin{figure}[h!]
\centering
\includegraphics[scale=0.35]{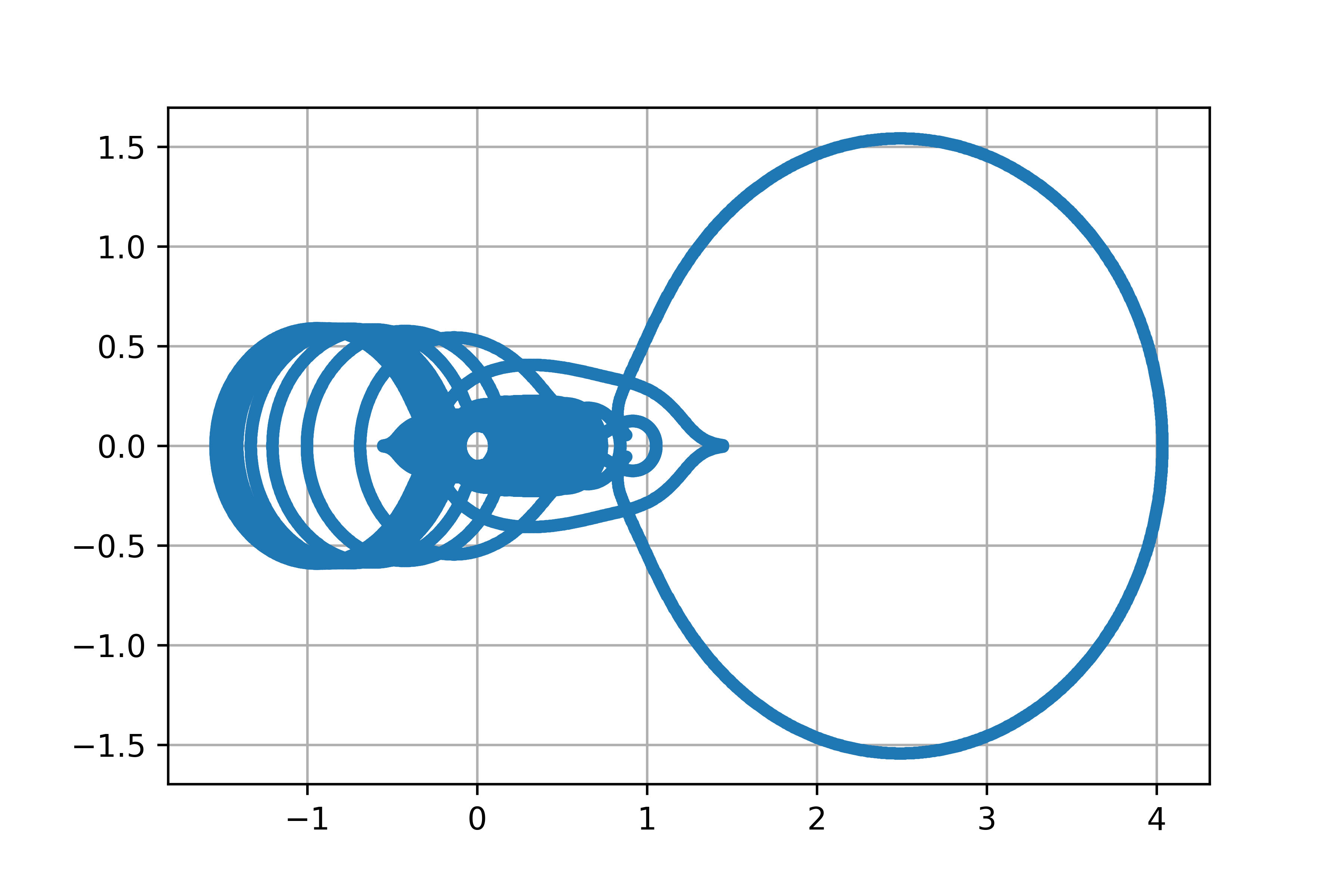}
\caption{The mapping $f(z)$ applied to the complex unit circle over 30 iterations, after normalization.}\label{fig:30it}
\end{figure}

\section{Directed cycles}\label{scycle}

As a consequence of Theorem~\ref{count}, there are no directed 3-cycles in an
ILDT digraph unless there is one present in the initial graph. We generalize
this property in the following result. Note that our directed cycles in the
theorem are oriented, and so do not include directed 2-cycles.  However,
we are allowed to include bidirected edges (traversed in a single-direction) as
part of our directed cycles.
\begin{theorem}
For all $t\ge 0,$ the digraph $G_t$ contains an oriented directed cycle if and only if $G_{t+1}$ contains an oriented directed cycle.
\end{theorem}
\begin{proof}
The forward implication is immediate, so we focus on the reverse implication. Let $(a_1, \ldots, a_k)$ be a directed cycle of length $k$ in $G_{t+1}$.
Define a function $f:V(G_{t+1}) \rightarrow V(G_{t})$ that maps a clone in $G_{t+1}$ to its parent node, and acts as the identity mapping, otherwise.
If $(a_i,a_{i+1})$ is an arc in $G_{t+1}$, then $(f(a_i), f(a_{i+1}))$ is an arc in $G_t$. The subgraph induced by the
edges of the closed directed walk $(f(a_1),\ldots, f(a_k))$ has
in-degree equal to out-degree at every node (counting the
multiplicity of the edges in the walk) precisely because it is a
closed directed walk.  Each time we visit a node on the walk we
contribute one to the in-degree and one to the out-degree. Hence, $(f(a_1),\ldots, f(a_k))$ decomposes into an edge-disjoint collection of directed cycles (none of which is a directed 2-cycle, by hypothesis). Hence, $G_t$ contains a cycle. \qed
\end{proof}

We turn next to directed Hamiltonian cycles; that is, directed cycles visiting each node exactly once. Note that while we do not expect directed Hamiltonian cycles in real-world social networks, the emergence of this property in ILDT graphs is of graph-theoretical interest in its own right. For this, we first prove the following theorem on Hamiltonian paths in ILT undirected graphs. For a graph
$G$, we use the notation $\ILT_t(G)$ for the ILT graph resulting at time $t$ if $G_0=G;$ analogous notation is used for ILDT graphs. We use the notation $G[S]$ for the subgraph induced by nodes $S$ in $G$.

In the following lemma with $G_0$ chosen as $K_1,$ we label the node of $G_0$ as $0,$ and its unique child in $G_1$ as $1.$
\begin{lemma}
Fix $t\ge 1$ and let $G_t=\ILT_t(K_1)$.  For every clone $v\in V(G_t)$, there is a Hamiltonian path in $G_t$ from $v$ to 0.
\label{path-lem}
\end{lemma}
\begin{proof}
We use induction on $t \ge 1$.  The base case $t=1$ is straightforward, since $G_1\cong K_2$.
For the induction step, we assume $t\ge 2$. We label the node of $G_0$ as $0,$ and its unique child in $G_1$ as $1.$  Note that we can partition $V(G_t)$
into $V_0$ and $V_1$ such that $0\in V_0$, $1\in V_1$ and $G_t[V_0]\cong
G_t[V_1]\cong G_{t-1}$. To see this, consider the ILDT process as starting with each of the nodes $0$ and $1$ independently. For $i=1,2$, the set $V_i$ consists of all clones over subsequent time-steps starting with the initial vertex $i$.

First suppose that $v\in V_1$, and choose an arbitrary $w\in V_0$. By the
induction hypothesis, there exists a Hamiltonian path $P_0$ in $G_t[V_0]$ with
endpoints $w$ and $0$.  Similarly, there exists a Hamiltonian path $P_1$ in
$G_t[V_1]$ with endpoints $v$ and $1$. Let $P=P_1+1w+P_0$; now $P$ is the
desired Hamiltonian path in $G_t$ from $v$ to 0.

Suppose instead that $v\in V_0$, and choose an arbitrary $w\in V_1$. By the
induction hypothesis, there exists a Hamiltonian path $P_0$ in $G_t[V_0]$ with
endpoints $v$ and $0$. Similarly, there exists a Hamiltonian path $P_1$ in
$G_t[V_1]$ with endpoints $w$ and $1$. Let $x$ be the neighbor of 0 on $P_0$.
Let $P=P_0-x0+x1+P_1+w0$.  Now $P$ is the desired Hamiltonian path in $G_t$ from
$v$ to 0. \qed
\end{proof}

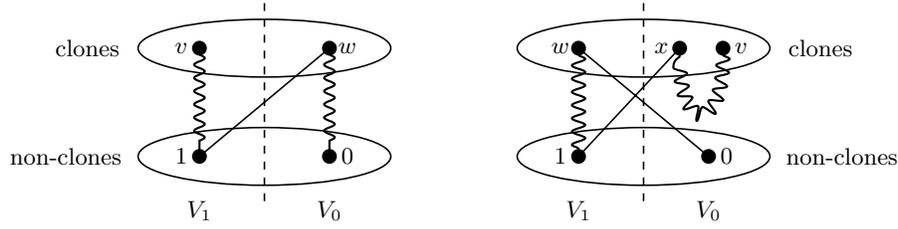
\begin{figure}[!h]
\centering
\begin{tikzpicture}[scale=.48,semithick]
\tikzstyle{usStyle}=[shape = circle, minimum size = 5.0pt, inner sep = 0pt,
outer sep = 0pt, draw, fill=black]

\draw (0,0) ellipse (3.5cm and .8cm); 
\draw (-4.9,0) node {\footnotesize{clones}};
\draw (0,-3) ellipse (3.5cm and .8cm); 
\draw (-5.5,-3) node {\footnotesize{non-clones}};
\draw (-1.8,0) node[usStyle] (v) {} ++ (-.5,0) node (vlab) {\footnotesize{$v$}};
\draw (1.8,0) node[usStyle] (w) {} ++ (.5,0) node (wlab) {\footnotesize{$w$}};
\draw (-1.8,-3) node[usStyle] (v1) {} ++ (-.5,0) node (v1lab) {\footnotesize{$1$}};
\draw (1.8,-3) node[usStyle] (v0) {} ++ (.5,0) node (v0lab) {\footnotesize{$0$}};

\draw[thick,snake=coil,segment aspect=0, segment amplitude=2pt,segment
length=5pt] (v) -- (v1) (w) -- (v0);
\draw (w) -- (v1);

\draw[dashed] (0,1.25) -- (0,-4.25); 
\draw (-1.8,-4.5) node (V1) {\footnotesize{$V_1$}};
\draw (1.8,-4.5) node (V0) {\footnotesize{$V_0$}};

\begin{scope}[xshift=10.5cm]
\draw (0,0) ellipse (3.5cm and .8cm); 
\draw (4.9,0) node {\footnotesize{clones}};
\draw (0,-3) ellipse (3.5cm and .8cm); 
\draw (5.5,-3) node {\footnotesize{non-clones}};
\draw (-1.8,0) node[usStyle] (w) {} ++ (-.5,0) node (wlab) {\footnotesize{$w$}};
\draw (2.2,0) node[usStyle] (v) {} ++ (.5,0) node (vlab) {\footnotesize{$v$}};
\draw (1.0,0) node[usStyle] (x) {} ++ (-.5,0) node (xlab) {\footnotesize{$x$}};
\draw (-1.8,-3) node[usStyle] (v1) {} ++ (-.5,0) node (v1lab) {\footnotesize{$1$}};
\draw (1.8,-3) node[usStyle] (v0) {} ++ (.5,0) node (v0lab) {\footnotesize{$0$}};

\draw[thick,decorate, decoration=snake, segment length=5pt]
(v) .. controls ++(270:2.25) and ++(270: 2.25) ..  (x);
\draw[thick,decorate, decoration=snake, segment length=5pt] (w) -- (v1);
\draw (x) -- (v1) (w) -- (v0);

\draw[dashed] (0,1.25) -- (0,-4.25); 
\draw (-1.8,-4.5) node (V1) {\footnotesize{$V_1$}};
\draw (1.8,-4.5) node (V0) {\footnotesize{$V_0$}};

\end{scope}
\end{tikzpicture}
\caption{The two cases in the proof of Lemma~\ref{path-lem}, wavy lines represent paths.
\label{path-lem-fig}}
\end{figure}

In a digraph $D$, a closed spanning walk $\CC$ (that respects the orientations
of $D$) is \emph{nice} if for each edge
$v_iv_j\in E(\CC)$ either (i) $v_iv_j$ is the last edge departing $v_i$ on $\CC$ or
(ii) $v_iv_j$ is the first edge entering $v_j$ on $\CC$; possibly both (i) and (ii)
hold for some edges of $\CC$. The \emph{max frequency}, written
$s(\CC)$, of a nice walk $\CC$ is the largest number of times that any node
appears in $\CC$.

\begin{theorem}
If $D$ is a digraph with a nice walk $\CC$ and $t\in\Z^+$ such that $2^{t-1}\ge
s(\CC)$, then $D_t$ has a directed Hamiltonian cycle.
\label{mainthm}
\end{theorem}
\begin{proof}
Let $D_t=\mathrm{ILDT}_t(D)$. We construct a Hamiltonian cycle in $D_t$ by the algorithm below.  We assume that
the nodes of $D$ are $\{v_1,\ldots,v_n\}$ (each $v_i$ may appear many times on
$\CC$) and that the nodes of $D_t$
are partitioned into $V_1,\ldots, V_n$ (where $V_i$ consists of $v_i$ and all
its \emph{descendants}; that is, nodes that resulted by iterated cloning of $v_i$).
Intuitively, we use $\CC$ to ensure that our walk $\HH$ visits each $V_i$ at
least once and we use each $P_i$ to ensure that we visit all remaining vertices of
$V_i$ the last time that $\CC$ visits $v_i$, in $D$.

\noindent \textbf{Initialization:} Pick an arbitrary node $v_i$ on $\CC.$
Choose a clone $w\in V_i$.  Start $\HH$ at $w$.
Let $P_i$ be a Hamiltonian path in $V_i$ that starts at $w$ and ends at 0 (in
$V_i$), by Lemma~\ref{path-lem} (with vertex 0 defined as in
Lemma~\ref{path-lem}).  Always $v_i, v_j$ refers to vertices of $D$ and
vertices of $D_t$ are denoted by $w$, $x$, or 0.

\noindent \textbf{Iteration:}
Assume that $\HH$ currently ends at some clone $w\in V_i$ (for some $i$) and
path $P_i$ is defined, possibly from the initialization.
If we have followed all edges of $\CC$, then halt and output $\HH$.  Otherwise,
let $e=v_iv_j$ denote the next edge of $\CC$.
(1) If $e$ is the last edge leaving $v_i$ on $\CC$, then follow $P_i$ from
$w$ to $0$ in $V_i$; otherwise, move to the next node on $P_i$. (2)
If $P_j$ is undefined (we have not yet visited $v_j$ on $\CC$),
then follow an edge to an arbitrary clone $x$ in $V_j$.
In this case, define $P_j$ to be a Hamiltonian path in $V_j$ with endpoints $x$
and 0; such a $P_j$ exists by Lemma~\ref{path-lem}.
If $P_j$ is defined, then follow an edge from the current node of $\HH$ to the next
node on $P_j$ (in $V_j$). (When $e$ is the final edge of $\CC$, we return to
the node of $V_j$ where we started, which finishes $\HH$.)

This completes the algorithm for constructing $\HH$ from $\CC$.
We prove its correctness in two steps.  First, we show that if the
algorithm completes, then it constructs the desired Hamiltonian cycle $\HH$.
Second, we show that the algorithm does indeed complete. Suppose the algorithm completes successfully.
Since $\CC$ is a spanning walk, it visits every node $v_i\in V(D)$. Thus,
$\HH$ visits every $V_i$.  The final time that $\HH$ visits a $V_i$ it visits
every remaining node on $P_i$.  Thus, $\HH$ visits every node in $\bigcup
V_i=V(D_t)$; that is, $\HH$ is spanning in $D_t$.

Now we must show that the algorithm completes successfully. Every time $\HH$
leaves a $V_i$ it does so from a non-clone, and every time $\HH$ returns
to a $V_j$ it returns to a clone (that has not been visited before).  The number
of clones in each $V_j$ is $2^{t-1}$, so this is possible precisely because
$s(\CC)\ge 2^{t-1}$.  Now we need to check that $D_t$ has the necessary edges
between $V_i$ and $V_j$.  Since $\CC$ is nice, each edge $e=v_iv_j\in E(\CC)$
satisfies that either (i) $v_iv_j$ is the last edge leaving $v_i$ on $\CC$ or
(ii) $v_iv_j$ is the first edge entering $v_j$ on $\CC$. In (i), $\HH$ leaves $V_i$
from node 0, which has edges to every node of $V_j$. In (ii), any edge from
a non-clone of $V_i$ to a clone of $V_j$ suffices, since we will define $P_j$ as
starting from this clone of $V_j$  (and since $\CC$ has never before visited
$v_j$).
\qed \end{proof}

The following result on the Hamiltonicity of the ILT model was first proven in \cite{ilm}, and we give an alternative proof as a corollary of Theorem~\ref{mainthm}.

\begin{corollary}
If $G$ is a connected undirected graph and $t=\log |V(G)|$, then $\ILT_t(G)$ is Hamiltonian.
\end{corollary}
\begin{proof}
We form a digraph $D$ from $G$ by replacing each undirected edge $vw$ with arcs
$(v,w)$ and $(w,v)$. We construct a nice spanning closed walk of $D$ and apply
Theorem~\ref{mainthm}. Choose an arbitrary node $v\in D$ and form $\CC$
by recording each edge followed in a depth-first traversal of $D$ (including to what we call
back-tracking edges). Consider an edge $vw\in E(\CC)$.  If $w$ has never been
visited before, then $vw$ satisfies (ii) in the definition of nice.  If
$w$ has been visited before, then it is straightforward to check that $vw$ satisfies (i)
in the definition (precisely because $\CC$ arose from a depth-first traversal of
$D$). \qed \end{proof}

We conjecture that for every strongly connected digraph $D$ there exists an integer $t$ such that
$\ILT_t(D)$ has a Hamiltonian cycle. In a sense, this conjecture is best possible, since if $\ILT_t(D)$ is Hamiltonian for some $t$, then $D$ must
be strongly connected.  Namely, if there exist $v_i,v_j\in V(D)$ such that $D$
has no directed path from $v_i$ to $v_j$, then $\ILT_t(D)$ has no directed path
from $V_i$ to $V_j$, so $\ILT_t(D)$ is not Hamiltonian.  We suspect that this
conjecture can be proved by somehow modifying the proof of
Theorem~\ref{mainthm}.

\section{Conclusion and further directions}

We introduced and analyzed the Iterated Local Directed Transitivity (ILDT) model for social networks, motivated by status theory, transitivity in triads, and the ILT model in the undirected case \cite{ilt}. We proved that the ILDT model, as in social networks, generates graphs that densify over time. A count of the directed, transitive, and bidirectional 3-cycles was given, and it was shown that the 3-transitive cycles count may be far more abundant by choice of the initial graph of the model. We studied the eigenvalues of the adjacency matrices of ILDT graphs, with a discussion of the limiting distribution of eigenvalues of the directed 3-cycle. We concluded our results with an analysis of directed cycles in ILDT graphs and proved that in many instances of the initial graph, ILDT graphs have Hamiltonian cycles.

Given our limited space, we did not explore distance properties of the model, although we expect the model should generate small-world graphs, as is the case for ILT graphs. In the full version of the paper, it would be interesting to analyze the clustering coefficient, domination number, and degree distribution of ILDT graphs. The eigenvalues of ILDT graphs are worthy of further study, both in their limiting distribution in the complex plane and regarding their spectral expansion.

\end{document}